\documentclass[11pt,oneside]{amsart}
\usepackage{amsmath,amssymb,psfrag}

\usepackage[dvips]{graphicx}

\usepackage{natbib}

\theoremstyle{plain} 
\newtheorem{theorem}{Theorem}[section]
\newtheorem{lemma}[theorem]{Lemma}
\newtheorem{corollary}[theorem]{Corollary}

\newtheorem{remark}[theorem]{Remark}

\numberwithin{equation}{section}

\begin{document}

\begin{center}

{\bf\large Duality, Ancestral and Diffusion Processes in Models with
Selection${}^*$}

\vspace{1cm}

Shuhei Mano ${}^\dagger$

{\it Department of Statistics, University of Oxford, OX1 3TG, 
United Kingdom\\
Graduate School of Natural Sciences, Nagoya City University, 
Nagoya 467-8501, Japan\\
Japan Science and Technology Agency, 4-1-8 Honcho, 
Kawaguchi 332-0012, Japan}

\end{center}

\vspace{2cm}

${}^*$ Part of the research for this paper was carried out while 
the author visited the Department of Statistics at University of Oxford.

${}^\dagger$ The author was supported in part by Grants-in-Aids for 
Scientific Research from the Ministry of Education, Culture, Sports, 
Science and Technology of Japan.

\vspace{1cm}

\noindent
Address for correspondence: Graduate School of Natural Sciences, 
Nagoya City University, 1 Yamanohata, Mizuho-cho, Mizuho-ku, Nagoya 
467-8501, Japan;\\ E-mail: {\tt mano@nsc.nagoya-cu.ac.jp}; Tel/Fax: 
+81-(0)52-872-5821




\newpage

\noindent
{\bf abstract}

\noindent
The ancestral selection graph in population genetics was introduced by 
\citet{KroneNeuhauser1997} as an analogue of the coalescent genealogy
of a sample of genes from a neutrally evolving population. The number of 
particles in this graph, followed backwards in time, is a birth and death 
process with quadratic death and linear birth rates. In this paper 
an explicit form of the probability distribution of the number of
particles is obtained by using the density of the allele frequency in 
the corresponding diffusion model obtained by \citet{Kimura1955}. It is 
shown that the process of fixation of the allele in the diffusion model 
corresponds to convergence of the ancestral process to its stationary 
measure. The time to fixation of the allele conditional on fixation is 
studied in terms of the ancestral process.

\vspace{1cm}


\newpage

\section{Introduction}

In population genetics, the genealogy of a sample of genes plays 
an important role in a probabilistic description of the sample.
Consider a discrete-time Wright-Fisher model of a population consisting
of $2N$ neutral genes. If we measure time in units of $2N$ generations and 
let $N$ tend to infinity, then the Wright-Fisher process converges to 
a diffusion process. The genealogy of a sample of $n$ genes, followed
backwards in time, is described by the coalescent process \citep{Kingman1982b}.
The convergence is robust under a number of different models 
(e.g. Moran model). Let $\{a_n(t);t\ge 0\}$ be the number of ancestors of 
a sample of $n$ genes at time $t$ in the past. Then the size process 
$\{a_n(t);t\ge 0\}$ is a death process with death rate $i(i-1)/2$ when 
the size is $i$. The size process will be referred to as the ancestral 
process. The distribution of $a_n(t)$ is known (see \citet{Griffiths1979}, 
\citet{Tavare1984}) and
\begin{equation}
{\mathbb P}[a_n(t)=i]=
\sum_{k=i}^{n}
\frac{(-1)^{k-i}(2k-1)(i)_{k-1}[n]_k}{i!(k-i)!(n)_k}\rho^0_k(t), 
\qquad i=1,2,...,n,
\label{neu_ni}
\end{equation}
where $\rho^0_k(t):=\exp\{-k(k-1)t/2\}, [n]_k=n(n-1)\cdots(n-k+1)$ and
$(n)_k=n(n+1)\cdots(n+k-1)$. The total variation norm between $a_n(t)$ and 
$a_n(\infty)=\delta_1$ has a simple form
\begin{equation}
\|a_n(t),\delta_1\|_{var}=1-{\mathbb P}[a_n(t)=1].
\label{neu_tv}
\end{equation}
There is a first time $W^0_{n,1}$ such that $a_n(W^0_{n,1})=1$, which is
the time to the most recent common ancestor. The density of $W^0_{n,1}$ 
follows
\begin{equation}
{\mathbb P}[W^0_{n,1}\le t]={\mathbb P}[a_n(t)=1]=
\sum_{k=1}^{n}
\frac{(-1)^{k-1}(2k-1)[n]_k}{(n)_k}\rho^0_k(t).
\label{neu_w}
\end{equation}
A bound for ${\mathbb P}[a_n(t)=1]$ is known 
(see \citet{Kingman1982b}, \citet{Tavare1984}) and
\begin{equation}
e^{-t}\le 1-{\mathbb P}[a_n(t)=1] \le 3\frac{n-1}{n+1} e^{-t}, 
\qquad n=2,3,...
\label{neu_bound}
\end{equation}

The ancestral selection graph introduced by {\citet{KroneNeuhauser1997}
is an analogue of the coalescent genealogy. Assume that a pair of 
allelic types $A_1$ and $A_2$ are segregating in a population, and 
the selective advantage of a type $A_1$ gene over a type $A_2$ gene is 
$s~(>0)$. Let $N\rightarrow\infty$ while $c=Ns$ is held constant. 
The elements are referred to as particles. Let $b_n(t)$ be the number 
of edges, or ancestral particles, in a cross section of the ancestral 
selection graph for a sample of $n$ genes at time $t$ in the past. 
In the ancestral selection graph, coalescence occurs at rate
$\alpha_i=i(i-1)/2$ and branching occurs at rate $\beta_i=2ci$ when 
the size is $i$. Then the ancestral process $\{b_n(t);t\ge 0\}$ is 
a birth and death process with rates $\beta_i$ and $\alpha_i$. A particle 
is called real if it is a part of the real genealogy of the sample, 
otherwise the particle is called virtual. If two particles reach 
a coalescing point, the resulting particle is real if and only if 
at least one of the two particles is real, otherwise the resulting 
particle is virtual. If a real particle reaches a branching point,
it splits into a real particle and into a virtual particle. 
If a virtual particle reaches a branching point, it splits into two
virtual particles. If a type $A_2$ particle reaches a branching point, 
it splits into two type $A_2$ particles. If a type $A_1$ particle reaches 
a branching point, it splits into two particles, where at least one of 
the two particles is type $A_1$. Because the death rates are quadratic 
while the the birth rates are only linear, there is a first time 
$W^c_{n,1}$ such that $b_n(W^c_{n,1})=1$. \citet{KroneNeuhauser1997} 
consider stopping the process at this time, since the genetic composition 
of the sample is determined by then. They called the ancestral particle 
at the time the ultimate ancestor. In the case of no mutation, the real 
genealogy embedded in an ancestral selection graph is the same 
as in the neutral process (\citet{KroneNeuhauser1997}; Theorem 3.12).
Thus the ancestral process of the real particles can be described by 
the neutral process $\{a_n(t);t\ge 0\}$.
In this article, we discuss properties of the ancestral process 
$\{b_n(t);t\ge 0\}$ without mutation which is not stopped upon reaching 
the ultimate ancestor. \citet{Fearnhead2002} has studied a process which 
is not stopped upon reaching the ultimate ancestor. He identifies 
the stationary distribution of the process and uses the distribution to 
characterize the substitution process to the common ancestor.

\citet{Kimura1955} studied the density of the allele frequency by analyzing 
the diffusion process to which the Wright-Fisher model with directional
selection converges. Let $x_p(t)$ be the frequency of the allele $A_1$ 
at time $t$ forward in a population in which the initial frequency is 
$p$. Then the Kolmogorov forward equation for the diffusion process 
$\{x_p(t); t\ge 0\}$ on $(0,1)$ is 
\begin{equation}
\frac{\partial \phi}{\partial t}=
\frac{1}{2}\frac{\partial^2}{\partial x^2}
\left\{
x(1-x)\phi
\right\}-2c\frac{\partial}{\partial x}
\left\{x(1-x)\phi\right\},
\end{equation}
with the initial condition $\phi(p,x;0)=\delta(x-p)$. The solution is
\begin{equation}
\phi(p,x;t)=
2(1-r^2)e^{c(r-1)}e^{2cx}
\sum_{k=0}^{\infty}\frac{V_{1k}^{(1)}(c,r)
V_{1k}^{(1)}(c,z)}{N_{1k}}\rho^c_{k+2}(t),
\label{sol_Kimura}
\end{equation}
where $r=1-2p,~z=1-2x,~\rho^c_{k+2}(t):=\exp(-\lambda_k t),~k=0,1,2,...$
and $-\lambda_k~(0<\lambda_0<\lambda_1<\cdots)$ are the eigenvalues of 
the generator. A plot of $\lambda_0$ is given in Figure 1 of 
\citet{Kimura1955}.
$V^{(1)}_{1k}(c,z)$ is the oblate spheroidal wave function (see Appendix):
\begin{equation}
V^{(1)}_{1k}(c,z)=\sum_{l\ge 0}{}'f^k_l(c) T^1_l(z),
\end{equation}
where $T^1_l(z)$ is the Gegenbauer function (may also be denoted by
$C^{\frac{3}{2}}_l(z)$) and the summation is over even values of 
$l$ if $k$ is even, odd values of $l$ if $k$ is odd. 
$N_{1k}$ is the normalization constant of $V^{(1)}_{1k}(c,z)$. 
The probability mass at the exit boundaries are 
\begin{equation}
f(p,1;t)=
2(1-r^2)e^{c(r+1)}\sum_{k=0}^{\infty}
\frac{V_{1k}^{(1)}(c,r)V_{1k}^{(1)}(c,-1)}{2\lambda_kN_{1k}}
(1-\rho^c_{k+2}(t)),
\label{fix_A1}
\end{equation}
and
\begin{equation}
f(p,0;t)=
2(1-r^2)e^{c(r-1)}\sum_{k=0}^{\infty}
\frac{V_{1k}^{(1)}(c,r)V_{1k}^{(1)}(c,1)}{2\lambda_kN_{1k}}
(1-\rho^c_{k+2}(t)).
\label{fix_A2}
\end{equation}

In Section 2, the ancestral process $\{b_n(t);t\ge 0\}$ without absorbing
states is studied. An explicit form of the probability distribution of
$b_n(t)$ is obtained by using the moment duality between the ancestral 
process and the Wright-Fisher diffusion with directional selection. 
It is also observed that the distribution of the ancestral process 
converges to a stationary distribution. In Section 3, the rate of convergence 
is discussed. In contrast to the neutral process, the final rates of 
convergence are given by the largest eigenvalue for all the states. 
Bounds for the probability that $b_n(t)$ is at the state 1 are obtained by 
an elementary martingale argument, which corresponds to the bounds 
(\ref{neu_bound}) for the neutral process. In Section 4, the ancestral 
process with absorbing states is considered. It is shown that the first 
passage times of the ancestral process $\{b_n(t);t\ge 0\}$ at the states 
$1,2,...,n-1$ are larger than that in the neutral process for all the states. 
By killing the modified process, in which the state 1 is the absorbing state, 
the form of the joint probability generating function of $b_n(t)$ and 
the number of branching events is obtained. By using this formula, 
the expectation of the total length of the edges in the ancestral selection 
graph is obtained. In Section 5, the ancestral process of the whole population 
$\{b_\infty(t);t\ge 0\}$ is studied. It is shown that the process of fixation 
of the allele in the diffusion model corresponds to convergence of 
the ancestral process to its stationary measure. The time to fixation of
an allele conditional on fixation is studied in terms of the ancestral 
process. It is shown that the density of the time to fixation of a single 
mutant gene conditional on fixation is given by the probability of the whole 
population being descended from a single real ancestral particle, regardless
of the allelic type. In the neutral process, the density of the waiting time 
until the ancestral process hits the state 1 and the density of 
the conditional fixation time are given by the probability that the ancestral 
process is at the state 1. The property does not hold in the process with 
selection.

\section{Number of ancestral particles}

In this section, we will obtain an equation that relates the moments of 
the Wright-Fisher diffusion with directional selection to a Markov process 
that specifies the number of particles (real and virtual) that are present 
in the ancestral selection graph. To derive this result, we will exploit 
the concept of duality from the theory of Markov processes 
(\citet{EK1986}; Section 4.4). If $X=\{X_t;t\ge 0\}$ and $Y=\{Y_t;t\ge 0\}$ 
are Markov processes with values in sets $Z_X$ and $Z_Y$, respectively, then 
$X$ and $Y$ are said to be dual with respect to a function $f(x,y)$ if 
the identity
\begin{equation}
{\mathbb E}_x[f(X_t,y)]={\mathbb E}_y[f(x,Y_t)]
\label{dual0}
\end{equation}
holds for every $x \in Z_X$ and $y \in Z_Y$. Duality is a useful concept 
because it allows us to use our knowledge of one process to learn about 
the other. Although there is no general procedure for identifying dual 
processes, duality can sometimes be deduced using simple generator 
calculations. Specifically, if $G_x$ is the infinitesimal generator of 
the process $X$ and $G_y$ is the infinitesimal generator of the process $Y$,
then the duality relationship shown in (\ref{dual0}) will be satisfied if 
the identity
\begin{equation}
G_x f(x,y)=G_y f(x,y)
\end{equation}
holds for all $x$ and $y$. Here we think $G_xf(x,y)$ as acting on the
$x$-variable of the function $f(x,y)$ for each fixed value of $y$.

To apply these results to the Wright-Fisher diffusion with selection, it 
will be necessary to consider the frequency $y_q(t)=1-x_p(t) (q=1-p)$ of 
the less fit allele, which is itself governed by a Wright-Fisher diffusion 
with generator:
\begin{equation}
G_yf(y)=\frac{1}{2}y(1-y)f''(y)-2cy(1-y)f'(y).
\end{equation}
Notice that the selection coefficient is negative in this case ($c\geq 0$). 
If we define the function $f(y,n)=y^n$, then a simple calculation shows that
\begin{eqnarray}
G_yf(y,n)&=&{n \choose 2}[f(y,n-1)-f(y,n)]+2cn[f(y,n+1)-f(y,n)]\nonumber\\
&=&G_nf(y,n),
\end{eqnarray}
where $G_n$ is the operator defined by
\begin{equation}
G_nf(n)={n\choose 2}[f(n-1)-f(n)]+2cn[f(n+1)-f(n)].
\end{equation}
In other words, $G_n$ is the infinitesimal generator of the birth and death
process $\{b_n(t);t\ge 0\}$ which keeps track of the number of ancestral 
particles in the ancestral selection graph. Because $f(y,n)$ is bounded, 
we can use a result of \citet{EK1986} (Chapter 4, Corollary 4.13) to deduce 
that the Wright-Fisher diffusion $\{y_q(t);t\ge 0\}$ and the birth and death 
process $b_n(t)$ are dual with respect to the function $f(y,n)$:
\begin{theorem}\label{theorem:duality}
\begin{equation}
{\mathbb E}[q^{b_n(t)}]={\mathbb E}[(y_q(t))^n], \qquad n=1,2,...
\label{dual1}
\end{equation}
\end{theorem}
Because the right-hand side of this identity involves moments of the process 
$\{y_q(t);t\ge 0\}$, the process $\{b_n(t);t\ge 0\}$ is said to be a 
moment dual for $\{y_q(t);t\ge 0\}$. The existence of moment duals of 
Wright-Fisher diffusions with polynomial coefficients was first shown by 
\citet{Shiga1981}, and the explicit description of duality between 
the birth and death process and the Wright-Fisher diffusion with directional 
selection was discussed by \citet{AS2005}. The duality can be proved in 
terms of the ancestral selection graph (\citet{AS2005}):

\begin{proof}
Partition an ancestral selection graph ${\mathcal G}$ into disconnected 
subgraphs 
${\mathcal G}_i, i=1,2,...$ Let ${\mathcal E}_t$ be the edges,
or the ancestral particles, of a cross section of ${\mathcal G}$ taken
at time $t$ 
backward. Then, $b_n(t)=|{\mathcal E}_t|$. Each
${\mathcal E}_0\cap{\mathcal G}_i$ consists only of type $A_2$ particles 
if and only if ${\mathcal E}_t\cap{\mathcal G}_i$ consists only of
type $A_2$ particles, since at least one type $A_1$ particle survives
from time $t$ to $0$ if ${\mathcal E}_t\cap{\mathcal G}_i$ contain
type $A_1$ particles. Here the ancestral selection graph is viewed
forward in time. If a type $A_1$ particle reaches a coalescence point, 
the number of type $A_1$ particles increase by 1. If a type $A_1$ particle 
reaches a branching point and meets another particle, the resulting particle 
is always type $A_1$. Thus, a sample consists only of type $A_2$ particles 
if and only if the ancestral particles at time $t$ consists only of type 
$A_2$ particles.
\end{proof}

If a sample contains type $A_1$ particles, then for $n=0,1,...; m=1,2,...$,
\begin{equation}
{\mathbb E}[(x_p(t))^m(y_q(t))^n]=
{\mathbb E}[(1-y_q(t))^m(y_q(t))^n]=
\sum_{i=0}^m
(-1)^i{m \choose i}{\mathbb E}[q^{b_{i+n}(t)}],
\label{dual_xy}
\end{equation}
with the convention $b_0(t)=0$. In particular,
\begin{equation}
{\mathbb E}[x_p(t)]={\mathbb E}[1-y_q(t)]={\mathbb E}[1-q^{b_1(t)}]
=\sum_{k=1}^\infty{\mathbb P}[b_1(t)=k]\sum_{i=1}^k
{k\choose i}p^iq^{k-i}.
\label{dual_x}
\end{equation}
The expression follows immediately from the distribution of the number of
particles in ${\mathcal E}_t$. Since the ancestral selection graph is 
irreducible, a type $A_1$ particle is sampled if and only if ${\mathcal E}_t$ 
contains at least one type $A_1$ particle. The first summation is over 
$|{\mathcal E}_t|$ and the second summation is over the number of 
type $A_1$ particles.

By the It\^{o} formula, we have a system of differential equations
for the moments of $y_q(t)$ 
\begin{equation}
\frac{d\xi_n}{dt}=
-(\alpha_n+\beta_n)\xi_n+\alpha_n\xi_{n-1}+\beta_n\xi_{n+1},
\qquad n=1,2,...
\label{moment}
\end{equation}
where $\xi_n=\mathbb{E}[(y_q(t))^n]$. The Kolmogorov backward equation
for the ancestral process $\{b_n(t);t\ge 0\}$ without absorbing states
is also given by (\ref{moment}), where $\xi_n={\mathbb P}[b_n(t)=i]$. 
Thus, (\ref{moment}) with $\xi_n={\mathbb E}[q^{b_n(t)}]$ holds. 
The isomorphism of these equations is a consequence of (\ref{dual1}).

A realization of the ancestral selection graph consists of two disconnected 
subgraphs embedded in a diagram of a sample path of $x_p(t)$ can be seen in 
Figure 1. The abscissa is the forward time interval $(0,t)$ and ordinate is 
$x_p(t)$. Thick lines represent the real genealogy. Lines used by type $A_2$
particles are dotted. The graph contributes to 
${\mathbb E}[(x_p(t))^3y_q(t)]$ and $b_4(t)=4$. If an ancestral selection 
graph consists of the upper subgraph only, it contributes to 
${\mathbb E}[y_q(t)]={\mathbb E}[q^{b_1(t)}]$ and $b_1(t)=2$.

Using an integral transform by the Gegenbauer function (\citet{Erdelyi1954})
for $l=0,1,...;n=1,2,...;i=0,1,...$,
\begin{equation}
\int_{-1}^1 T^1_l(z)(1+z)^n(1-z)^i dz=
\frac{2^{n+i}i!(l+1)(l+2)}{(n+1)_{i+1}}{}_3F_2(-l,l+3,i+1;2,i+n+2;1),
\end{equation}
where ${}_3F_2(-l,l+3,i+1;2,i+n+2;1)$ is the generalized hypergeometric
function, and with an identity (\ref{id_fix}), it is possible to obtain
an explicit expression for the probability generating function of $b_n(t)$, 
and we have
\begin{eqnarray}
{\mathbb E}[q^{b_n(t)}]&=&
{\mathbb E}[y_q(t)^n]=
\int^1_0 (1-x)^n\phi(x,p;t) dx+f(p,0;t)\nonumber\\
&=&\frac{e^{4c q}-1}{e^{4c}-1}
+2(1-r^2)e^{c(r-1)}
\sum^{\infty}_{k=0}
\frac{V_{1k}^{(1)}(c,r)}{N_{1k}}
\left\{F_{kn}(c)-\frac{V_{1k}^{(1)}(c,1)}{2\lambda_k}
\right\}\rho^c_{k+2}(t)
\nonumber\\
&=&\sum_{i=1}^{\infty} {\mathbb P}[b_n(t)=i] q^i,
\label{pgf_b}
\end{eqnarray}
where 
\begin{equation*}
F_{kn}(c):=\sum_{l\ge 0}{}'f^k_l(c)\sum_{i=0}^{\infty}
\frac{(2c)^i}{(n+1)_{i+1}}\sum_{j=0}^{l}{}
\frac{(-l)_j(l+1)_{j+2}(i+1)_j}{2\cdot j!(j+1)!(i+n+2)_j},
\end{equation*}
$r=1-2p$ and $f(p,0;t)$ is the probability mass at 0 given in (\ref{fix_A2}).
If $n=\infty$, $F_{k\infty}(c)=0$ and $f(p,0;t)$ gives the probability 
generating function. Using a power series expansion in $q$ of 
the Gegenbauer function
\begin{equation}
T^1_l(r)=(-1)^l
\sum_{i=0}^{l}\frac{(-l)_i(l+1)_{i+2}}{2\cdot i!(i+1)!}q^i,
\qquad l=0,1,..
\end{equation}
we obtain an explicit expression for the probability distribution of $b_n(t)$:
\begin{equation}
{\mathbb P}[b_n(t)=1]=
\pi_1+8e^{-2c}\sum^{\infty}_{k=0}
\frac{V_{1k}^{(1)}(c,-1)}
{N_{1k}}
\left\{F_{kn}(c)-\frac{V_{1k}^{(1)}(c,1)}{2\lambda_k}\right\}
\rho^c_{k+2}(t),
\label{sel_n1}
\end{equation}
and
\begin{equation}
{\mathbb P}[b_n(t)=i]=\pi_i+
8e^{-2c}\sum^{\infty}_{k=0}
\frac{G_{ki}(c)}
{N_{1k}}
\left\{F_{kn}(c)-\frac{V_{1k}^{(1)}(c,1)}{2\lambda_k}\right\}
\rho^c_{k+2}(t),\qquad i=2,3,...,
\label{sel_ni}
\end{equation}
where
\begin{eqnarray*}
G_{ki}(c)&:=&
\sum_{l\ge i-1}{}'f^k_l(c)(-1)^l
\frac{(-l)_{i-1}(l+1)_{i+1}}{2(i-1)!i!}\nonumber\\
&&
+\sum_{j=1}^{i-1}
\frac{(2c)^{j-1}(2c-j)}{j\cdot(j-1)!}
\sum_{l\ge i-j-1}{}'f^k_l(c)(-1)^l
\frac{(-l)_{i-j-1}(l+1)_{i-j+1}}{2(i-j-1)!(i-j)!},
\end{eqnarray*}
and $\pi_i$ are given in (\ref{stationary}). Note that there are finite 
probabilities at the states $n+1,n+2,...$. The expected number of ancestral
particles is
\begin{eqnarray}
{\mathbb E}[b_n(t)]&=&
\pi_1e^{4c}
+8e^{-2c}\sum^{\infty}_{k=0}
\frac{V_{1k}^{(1)}(c,-1)}{N_{1k}}
\left\{F_{kn}(c)-\frac{V_{1k}^{(1)}(c,1)}{2\lambda_k}\right\}
\rho^c_{k+2}(t)\nonumber\\
&&+8e^{-2c}\sum_{i=2}^{\infty}i\sum^{\infty}_{k=0}
\frac{G_{ki}(c)}{N_{1k}}
\left\{F_{kn}(c)-\frac{V_{1k}^{(1)}(c,1)}{2\lambda_k}\right\}
\rho^c_{k+2}(t),
\end{eqnarray}
and the falling factorial moments are 
\begin{equation}
{\mathbb E}[[b_n(t)]_i]=
i!\pi_ie^{4c}
+8e^{-2c}\sum_{j=i}^{\infty}[j]_i\sum^{\infty}_{k=0}
\frac{G_{kj}(c)}{N_{1k}}
\left\{F_{kn}(c)-\frac{V_{1k}^{(1)}(c,1)}{2\lambda_k}\right\}
\rho^c_{k+2}(t),\qquad i=2,3,...
\end{equation}

For small $c$, the probability distribution is approximately
\begin{equation}
{\mathbb P}[b_n(t)=1]
={\mathbb P}[a_n(t)=1]
-2c+2c\sum_{k=2}^{n+1}(-1)^k(2k-1)
\left\{\frac{[n]_{k}}{(n)_{k}}+
\frac{k(k-1)[n]_{k-1}}{(n)_{k+1}}\right\}\rho^0_{k}(t)
+O(c^2),
\end{equation}
and for $i=2,3,...$,
\begin{eqnarray}
{\mathbb P}[b_n(t)=i]
&=&{\mathbb P}[a_n(t)=i]+2c\delta_{i,2}
-2c\sum_{k=i}^{n+1}
\frac{(-1)^{k-i}(2k-1)(i)_{k-1}}{i!(k-i)!}
\left\{
\frac{k(k-1)[n]_{k-1}}{(n)_{k+1}}\right.
\nonumber\\
&&
\left.+\frac{(k^2-k+2i-2)[n]_{k}}{(k-i+1)(k+i-2)(n)_{k}}\right\}
\rho^0_{k}(t)
+2c\frac{(i-1)_{i-1}[n]_{i-1}}{(i-1)!(n)_{i-1}}
\rho^0_{i-1}(t)+O(c^2),
\end{eqnarray}
with a convention $[n]_{n+1}=0$, where $a_n(t)$ is the number of ancestors 
of a sample of $n$ neutral genes at the time $t$ in the past.

It is possible to obtain the solution of (\ref{moment}) as a perturbation
series in $2c$, where the series is represented by eigenvalues of the neutral 
process. Let 
\begin{equation}
\xi_n=\xi_n^{(0)}+2c\xi_n^{(1)}+(2c)^2\xi_n^{(2)}+\cdots,\qquad n=1,2,... 
\end{equation}
Denote the rate matrix of the neutral process $\{a_n(t);t\ge 0\}$ by 
$Q_0=(q_{0,ij})$, where $q_{0,i+1,i}=\alpha_{i+1}, q_{0,ii}=-\alpha_i$ for 
$i=1,2,...$ and other elements are zero. Let the Laplace transform of 
$\xi_n^{(i)}(t)$ be $\nu_n^{(i)}(\lambda)$. It is straightforward to show that
\begin{equation}
\nu^{(i)}=\{(Q_0-\lambda E)^{-1}C\}^i(\lambda E-Q_0)^{-1}\xi^{(0)}(0),
\qquad i=1,2,...
\end{equation}
where $C=(c_{ij})$ is given by $c_{ii}=i, c_{i,i+1}=-i$ for $i=1,2,...$
and other elements are zero. Note that the inverse Laplace transform of 
the element in the
$n$-th row and $i$-th column of the matrix 
$\{(Q_0-\lambda E)^{-1}C\}^j(\lambda E-Q_0)^{-1}$ gives the $j$-th order
coefficients in $2c$ of ${\mathbb P}[b_n(t)=i]$.    

Let $r_n(t)$ be the number of branching events in the time interval $(0,t)$
in the ancestral selection graph of a sample of $n$ genes, where $r_n(0)=0$.
The joint probability generating functions of $b_n(t)$ and $r_n(t)$
satisfy a system of differential equation
\begin{equation}
\frac{d\xi_n}{dt}=
-(\alpha_n+v\beta_n)\xi_n+\alpha_n\xi_{n-1}+v\beta_n\xi_{n+1}
-(1-v)\beta_n\xi_n,
\qquad n=1,2,..
\label{killing}
\end{equation}
with the initial condition $\xi_n(0)=q^n$, where 
$\xi_n={\mathbb E}[q^{b_n(t)}v^{r_n(t)}]$. The solution is given 
by killing of the modified process $\{\tilde{b}_n(t);t\ge 0\}$ in which
the selection coefficient is $vc$, and we have
\begin{equation}
{\mathbb E}[q^{b_n(t)}v^{r_n(t)}]=
{\mathbb E}\left[
q^{\tilde{b}_n(t)}\exp\left\{-2c(1-v)\int_0^t\tilde{b}_n(u)du\right\}\right].
\end{equation}
By setting $v=0$ in the last expression, we obtain the identity
\begin{equation}
{\mathbb E}[q^{b_n(t)},r_n(t)=0]=
{\mathbb E}
\left[
q^{a_n(t)}\exp\left\{-2c\int_0^ta_n(u)du\right\}
\right],
\end{equation}
which shows the Poisson nature of the branching in the ancestral 
process $\{b_n(t);t\ge 0\}$. The integral is the total length of the edges
in the neutral genealogy without branching in $(0,t)$. In particular,
${\mathbb P}[b_1(t)=1,r(t)=0]=e^{-2ct}$.

\section{Convergence}

Standard results on birth and death processes 
(see, e.g., \citet{KarlinTaylor1975}) 
gives the stationary measure of the ancestral process $\{b_n(t); t\ge 0\}$.
It is straightforward to obtain the stationary measure
\begin{equation}
\pi_i:={\mathbb P}[b_n(\infty)=i]=\frac{(4c)^i}{i!(e^{4c}-1)},
\quad i=1,2,...,
\label{stationary}
\end{equation}
which is the zero-truncated Poisson distribution.
Since the ancestral process of the real particles is the neutral process
$\{a_n(t);t\ge 0\}$ and the number of real particles becomes 1 in 
finite time, $\pi_1$ is the probability that there are no virtual 
particles. 

It is clear from (\ref{sel_n1}) and (\ref{sel_ni}) that for $i=1,2...$,
\begin{equation}
\pi_i-{\mathbb P}[b_n(t)=i]=O(\rho^c_2(t)),\qquad t\rightarrow\infty.
\label{sel_conv}
\end{equation}
For small $c$, the final rates of convergence are approximately
\begin{equation}
\lim_{t\rightarrow\infty}(\rho_2^{c}(t))^{-1}
\left\{\pi_1-{\mathbb P}[b_n(t)=1]\right\}
=3e^{-2c}
\left\{\frac{n-1}{n+1}-\frac{4c}{(n+1)(n+2)}\right\}
-2ce^{-2c}\delta_{n,1}+O(c^2),
\end{equation}
\begin{equation}
\lim_{t\rightarrow\infty}(\rho_2^{c}(t))^{-1}
\left\{\pi_2-{\mathbb P}[b_n(t)=2]\right\}
=-3e^{-2c}
\left\{\frac{n-1}{n+1}-\frac{2n(n-1)c}{n+2}\right\}
-2ce^{-2c}\delta_{n,1}+O(c^2),
\end{equation}
and 
\begin{equation}
\lim_{t\rightarrow\infty}(\rho_2^{c}(t))^{-1}
\left\{\pi_i-{\mathbb P}[b_n(t)=i]\right\}
=-3e^{-2c}\frac{n-1}{n+1}\frac{(2c)^{i-2}}{(i-2)!}+O(c^{i-1}),
\qquad i=3,4,...
\end{equation}
In contrast to the neutral process, the final rates of 
convergence are given by the largest eigenvalue for all the states. 
In the neutral process, we have
\begin{equation}
\lim_{t\rightarrow\infty}(\rho^0_i(t))^{-1}{\mathbb P}[a_n(t)=i]
=\frac{(i)_i[n]_i}{i!(n)_i},\qquad i=1,2,...,n.
\end{equation} 
The total variation norm has no simple form as in (\ref{neu_tv}).

A simple argument gives a bound for ${\mathbb P}[b_n(t)=1]$. 
The event that the number of ancestral particles is 1 is a subset of 
the event that the number of real particle is 1, and we have
\begin{equation}
{\mathbb P}[b_n(t)=1]\le{\mathbb P}[a_n(t)=1], \qquad n=1,2,...
\label{sel_bound_a}
\end{equation}

An elementary argument on a martingale gives explicit bounds for 
${\mathbb P}[b_n(t)=1]$, which corresponds to the bounds (\ref{neu_bound}})
in the neutral process. Let $\eta(n;c)$ satisfy a recursion
\begin{equation}
(\lambda_0-\alpha_n-\beta_n)\eta(n;c)+\alpha_n\eta(n-1;c)+
\beta_n\eta(n+1;c)=0,
\qquad n=1,2,...
\label{rec_eta}
\end{equation}
with the boundary condition $\eta(1;c)=-2c$. Since $\eta$ is an eigenvector
of the transition probability matrix of the ancestral process 
$\{b_n(t);t\ge 0\}$, $\eta(b_n(t);c)(\rho_2^{c}(t))^{-1}$ is a martingale 
to the ancestral process (see, e.g., \citet{KarlinTaylor1975}). Then,
\begin{equation}
{\mathbb E}[\eta(b_n(t);c)]=\eta(n;c)\rho_2^{c}(t).
\label{martingale}
\end{equation}
Although the explicit form of $\eta(n;c)$ is not available, it is possible 
to obtain an asymptotic form. Because of
\begin{equation}
\frac{\eta(n;c)}{\eta(n-1;c)}\rightarrow 
1+\frac{2\lambda_0}{n(n-1)}+O(n^{-3}),
\qquad 
n\rightarrow\infty,
\end{equation}
we deduce the asymptotic form
$\eta(n;c)\approx \eta(\infty;c)(1-2\lambda_0/n)$, where $\eta(\infty;c)$
is a function of $c$. Although the explicit form of $\eta(\infty;c)$ is 
not available, it is very close to $3\exp(-c)$ (See Figure 2). 

For small $c$, $\eta(n;c)$ can be expanded into a power series in $c$:
\begin{equation}
\eta(n;c)=3\frac{n-1}{n+1}
\left\{
1-\frac{n^2+n+2}{(n-1)(n+2)}c+\frac{4(3n^2+10n+18)}{25(n+2)(n+3)}c^2
\right\}+O(c^3),\qquad n=2,3,...
\label{etainf}
\end{equation}
Note that $\eta(\infty;c)$ is not exactly equal to $3\exp(-c)$.

\begin{lemma}\label{lemma:monotonic}
$\eta(n;c)$ is monotonically increasing in $n$.
\end{lemma}

\begin{proof}
Denote the rate matrix of the ancestral process $\{b_n(t);t\ge 0\}$ by 
$Q_c=(q_{c,ij})$, where 
$q_{c,i+1,i}=\alpha_{i+1}, q_{c,ii}=-(\alpha_i+\beta_i), q_{c,i,i+1}=\beta_i$
for $i=1,2,...$ and other elements are zero. $\eta$ is an eigenvector of 
an oscillatory matrix $E+Q_c(2N)^{-1}$ which belongs to the second largest 
eigenvalue $1-\lambda_0(2N)^{-1}$. An eigenvector of an oscillatory matrix 
which belongs to the second largest eigenvalue has exactly one variation of 
sign in the coordinates (see, \citet{Gantmacher1959}, pp. 105). Assume 
$\eta(i;c)>0,~i\ge L$ and $\eta(i;c)\le 0,~1\le i\le L-1$. 
Suppose $l\ge L-1$. By an induction we deduce from (\ref{rec_eta})
that
\begin{eqnarray*}
\eta(l+1;c)-\eta(l;c)
&=&\frac{\alpha_l}{\beta_l}
\{\eta(l;c)-\eta(l-1;c)\}-\lambda_0\eta(l;c)
\nonumber\\
&=&\frac{(l-1)!}{(4c)^{l-1}}
\left\{
\eta(2;c)-\eta(1;c)-\frac{\lambda_0}{\pi_2}
\sum_{i=2}^l\eta(i;c)\pi_i\right\}.
\end{eqnarray*}
By taking $t=\infty$ in (\ref{martingale}) it follows that
\begin{equation}
\sum_{i=1}^{\infty}\eta(i;c)\pi_i=0.
\label{etapi}
\end{equation}
Thus,
\begin{equation}
\eta(l+1;c)-\eta(l;c)=\frac{\lambda_0}{8c^2\pi_{l-1}}
\sum_{i=l+1}^\infty\eta(i;c)\pi_i>0.
\end{equation}
Next, suppose $2\le l\le L-2$. We have
\begin{equation}
\eta(l+1;c)-\eta(l;c)
=\frac{(l-1)!}{(4c)^{l-1}}
\left\{
\eta(2;c)-\eta(1;c)-\frac{\lambda_0}{\pi_2}
\sum_{i=2}^l\eta(i;c)\pi_i
\right\}>0.
\end{equation}
Finally, $\eta(2;c)-\eta(1;c)=\lambda_0>0$.
\end{proof}

From Lemma \ref{lemma:monotonic} it follows that
\begin{eqnarray}
{\mathbb P}[b_n(t)=1]\eta(1;c)+{\mathbb P}[b_n(t)>1]\eta(2;c)
&\le&{\mathbb E}[\eta(b_n(t);c)]\nonumber\\
&\le&{\mathbb P}[b_n(t)=1]\eta(1;c)
+{\mathbb P}[b_n(t)>1]\eta(\infty;c),
\end{eqnarray}
here we note that there are finite probabilities at the states larger
than $n$. Then, from (\ref{martingale}) we have the following bounds:

\begin{theorem}\label{theorem:bound}
If $\eta(n;c)$ satisfies the recursion (\ref{rec_eta}), then
\begin{equation}
\frac{\eta(n;c)e^{-\lambda_0 t}+2c}{\eta(\infty;c)+2c}\le
1-{\mathbb P}[b_n(t)=1]\le
\frac{\eta(n;c)e^{-\lambda_0 t}+2c}{\lambda_0},
\qquad n=1,2,...
\end{equation}
\end{theorem}

Figure 3 shows the bounds when $c=0,0.1,1$ and 8 for $n=10$. When $c>0$,
in contrast to the neutral case, the states larger than 1 have finite 
probabilities in the stationary measure (\ref{stationary}) and 
the upper and lower bounds do not converge to the same value.
Figure 3b shows that the upper bound is sharp for small values of $c$ 
and loose for intermediate values of $c$ (Figure 3c), and that both 
the upper and lower bounds converge to $1$ as $c$ tends to infinity.

\section{First passage times}

Let 
\begin{equation}
W^c_{n,i}:=\inf\{t\ge 0;b_n(t)=i\}, \qquad i=1,2,...
\end{equation}
and $\{b_{n}^1(t);W^c_{n,1}\ge t\ge 0\}$ be a modified process, 
where there is an absorbing state at 1, or the ultimate ancestor.
The modified 
process is the same as that introduced for ancestral recombination 
graph, where $4c$ is replaced by the recombination parameter $\rho$ 
(\citet{Griffiths1991}). Theorems 1, 2, 3 in \citet{Griffiths1991} 
hold for the modified process. The modified process was studied by 
\citet{KroneNeuhauser1997}. Here, modified processes
$\{b_{n}^i(t);W^c_{n,i}\ge t\ge 0\}$, where there is an absorbing state at
$i=1,2,...,n-1$, are studied to discuss the first passage times of 
the ancestral process $\{b_n(t);t\ge 0\}$ at the states $1,2,...,n-1$. 

It is possible to show that the expected first passage times of 
the ancestral process $\{b_n(t);t\ge 0\}$ at the states $1,2,...,n-1$ 
are larger than those in the neutral process $\{a_n(t);t\ge 0\}$.
${\mathbb E}[W^c_{n,1}]$ is given in \citet{KroneNeuhauser1997}.

\begin{theorem}\label{theorem:passage}
Let 
\begin{equation}
W^0_{n,i}:=\inf\{t\ge 0;a_n(t)=i\}, \qquad i=1,2,...
\end{equation}
Then,
\begin{equation}
{\mathbb E}[W^c_{n,i}]
=2\sum_{k=i}^{n-1}\sum_{j=0}^{\infty}\frac{(4c)^j}{(k)_{j+2}}
>{\mathbb E}[W^0_{n,i}],
\qquad i=1,2,...,n-1,
\end{equation}
where $W^c_{n,1}$ is the time to the ultimate ancestor.
\end{theorem}

\begin{proof}
The theorem follows from standard results on birth and death processes 
(see, e.g., \citet{KarlinTaylor1975}). The modified processes 
$\{b^i_n(t);W^c_{n,i}\ge t\ge 0\}$ hit the states $i=1,2,...,n-1$ in 
finite time with probability one, since
\begin{equation}
\sum^\infty_{m=i}\prod_{k=i+1}^{m+1}\frac{\alpha_k}{\beta_k}
=\frac{(4c)^{i-1}}{(i-1)!}\sum^\infty_{m=i}\frac{m!}{(4c)^m}=\infty,
\qquad i=1,2,...,n-1.
\end{equation}
From the Kolmogorov backward equation for the modified process
$\{b^i_n(t);W^c_{n,i}\ge t\ge 0\}$, which is (\ref{moment}) for 
$n=i+1,i+2,...$ with 
$\xi_n={\mathbb P}[b^i_n(t)=i]$ and the boundary condition 
$\xi_i=\delta(t)$, the expected first passage times 
satisfy a recursion for $i=1,2,...,n-1$
\begin{equation}
(\alpha_n+\beta_n)\zeta(n)-\alpha_n\zeta(n-1)-\beta_n\zeta(n+1)=1,
\qquad n=i+1,i+2,...,n-2
\end{equation}
with the boundary condition $\zeta(i)=0$, where 
$\zeta(n)={\mathbb E}[W^c_{n,i}]$. 
It is straightforward to solve the recursion and obtain
\begin{equation}
{\mathbb E}[W^c_{n,i}]=\sum_{m=i}^\infty\gamma_m+\sum_{j=i}^{n-2}
\prod_{k=i+1}^{j+1}\frac{\alpha_k}{\beta_k}\sum_{l=j+1}^\infty\gamma_l
=2\sum_{k=i}^{n-1}\sum_{j=0}^{\infty}\frac{(4c)^j}{(k)_{j+2}},
\qquad i=1,2,...,n-2,
\label{exp_Wnic1}
\end{equation}
and
\begin{equation}
{\mathbb E}[W^c_{n,n-1}]=\sum_{m=i}^\infty\gamma_m
=2\sum_{j=0}^{\infty}\frac{(4c)^j}{(k)_{j+2}},
\label{exp_Wnic2}
\end{equation}
where
\begin{equation*}
\gamma_{i}=\frac{1}{\alpha_{i+1}}=\frac{2}{i(i+1)},\qquad
\gamma_m=\frac{\beta_{i+1}\beta_{i+2}\cdots\beta_m}
{\alpha_{i+1}\alpha_{i+2}\cdots\alpha_m\alpha_{m+1}}
=\frac{2(4c)^{m-i}}{(i)_{m-i+2}},
\qquad m=i+1,i+2,...
\end{equation*}
It is clear from (\ref{exp_Wnic1}) and (\ref{exp_Wnic2}) that
\begin{equation}
{\mathbb E}[W^c_{n,i}]
>2\sum_{k=i}^{n-1}\frac{1}{k(k+1)}=2\left(\frac{1}{i}-\frac{1}{n}\right)
={\mathbb E}[W^0_{n,i}], \qquad i=1,2,...,n-1.
\label{ineq_wait}
\end{equation}
\end{proof}

As $c\rightarrow\infty$, for $i=1,2,...,n-1$,
\begin{equation}
{\mathbb E}[W^c_{n,i}]
=2\sum_{k=i}^{n-1}\frac{1}{k(k+1)}
\sum_{j=0}^{\infty}\frac{\left(\frac{4c}{k+2}\right)^j}
{\prod_{l=0}^{j-1}\left(1+\frac{l}{k+2}\right)}
>2\sum_{k=i}^{n-1}\frac{e^{\frac{4c}{k+2}}}
{k(k+1)}\rightarrow\infty.
\end{equation}

\begin{corollary}\label{corollary:passage}
For the whole population ($n=\infty$), the expected first passage times are
\begin{equation}
{\mathbb E}[W^c_{\infty,i}]
=2\sum_{j=0}^{\infty}\frac{(4c)^j}{(j+1)(i)_{j+1}},
\qquad
i=1,2,...
\end{equation}
\end{corollary}

\begin{proof}
It follows immediately from an identity
\begin{equation}
\sum_{k=i}^\infty\frac{1}{(k)_{j+2}}=\frac{1}{(j+1)(i)_{j+1}},
\qquad
j=0,1,...
\end{equation}
\end{proof}

It is straightforward to obtain higher moments of the first passage times of 
the ancestral process $\{b_n(t);t\ge 0\}$ at the states $1,2,...,n-1$
in the same manner. The second moments ${\mathbb E}[(W^c_{n,i})^2]$ 
satisfy a recursion
\begin{equation}
(\alpha_n+\beta_n)\zeta(n)
-\alpha_n\zeta(n-1)-\beta_n\zeta(n+1)
=2{\mathbb E}[W^c_{n,i}],\qquad n=i+1,i+2,...,n-2
\end{equation}
with the boundary condition $\zeta(i)=0$, where
$\zeta(n)={\mathbb E}[(W^c_{n,i})^2]$.
However, there is no simple form for the density as in (\ref{neu_w}). 
The Laplace transform of the first passage times of the ancestral 
process satisfy a recursion for $i=1,2,...,n-1$
\begin{equation}
(\lambda+\alpha_n+\beta_n)\zeta(n)-\alpha_n\zeta(n-1)-\beta_n\zeta(n+1)=0, 
\qquad n=i+1,i+2,...,n-2
\end{equation}
with the boundary condition $\zeta(i)=1$, where 
$\zeta(n)={\mathbb E}[e^{-\lambda W^c_{n,i}}]$. 

The joint probability generating function of $b^1_n(t)$ and $r_n(t)$
satisfies a system of differential equation (\ref{killing}) with
$\xi_n={\mathbb E}[q^{b^1_n(t)}v^{r_n(t)}]$. By taking $t=\infty$, we have 
\begin{equation}
0=-(\alpha_n+\beta_n)\xi_n+\alpha_n\xi_{n-1}+v\beta_n\xi_{n+1},
\qquad n=1,2,..,
\end{equation}
with the boundary condition $\xi_1=1$, where 
$\xi_n={\mathbb E}[v^{r_n(\infty)}]$. The form of the probability 
generating function of $r(\infty)$ is
\begin{equation}
{\mathbb E}[v^{r_n(\infty)}]=
{\mathbb E}
\left[
\exp\left\{-2c(1-v)\int_0^{W^{vc}_{n,1}}\tilde{b}_n^1(u)du\right\}
\right],
\label{pgf_r}
\end{equation}
while the explicit form of the probability generating function is given by 
Theorem 5.1 in \citet{EthierGriffiths1990}, where $\rho$ is replaced by 
$4c$, and we have
\begin{equation}
{\mathbb E}[s^{r_n(\infty)}]=\frac{R_n(v)}{R_1(v)},
\end{equation}
where
\begin{eqnarray*}
R_n(v)&=&\int_0^1x^{4c(1-v)-1}(1-x)^{n-1}e^{-4cv(1-x)}dx
\nonumber\\
&=&\frac{(n-1)!}{(4c(1-v))_{n}}{}_1F_1(n;4c(1-v)+n;-4cv)
\nonumber\\
&=&\sum_{i=0}^\infty
\frac{(n+i-1)!(-4cv)^i}{(4c(1-v)+n)_{n+i}i!}.
\end{eqnarray*}
(\ref{pgf_r}) provides a way to compute the expectation of 
the total length of the edges in the ancestral selection graph
in the time interval $(0,W_{n,1}^c)$, and we have
\begin{equation}
{\mathbb E}
\left[
\int_0^{W^c_{n,1}}b_n^1(u)du
\right]
=\frac{1}{2c}{\mathbb E}[r_n(\infty)]
=\sum_{k=1}^{\infty}(4c)^{k-1}\sum_{m=1}^{n-1}\frac{1}{(m)_k}.
\end{equation}

It is possible to obtain the probability that the modified process
$\{b^1_n(t);W^c_{n,1}\ge t\ge 1\}$ hits the states $n+1,n+2,...$

\begin{theorem}
Let $z(1)=0, z(2)=1$, and
\begin{equation}
z(j)=1+\frac{\alpha_2}{\beta_2}
+\frac{\alpha_2\alpha_3}{\alpha_2\beta_3}+
\cdots
+\frac{\alpha_2\alpha_3\cdots\alpha_{j-1}}{\beta_2\beta_3\cdots\beta_{j-1}}
=\sum_{k=0}^{j-2}\frac{k!}{(4c)^k},
\qquad j=3,4,...
\end{equation}
Then, the probability that the modified process
$\{b^1_n(t);W^c_{n,1}\ge t\ge 1\}$ hits the states $m=n+1,n+2,...$ is
\begin{equation}
{\mathbb P}[W^c_{n,1}>W^c_{n,m}]=\frac{z(n)}{z(m)}.
\end{equation}
\end{theorem}

\begin{proof}
The theorem follows from standard results on birth and death processes
(see, \citet{KarlinTaylor1975}, pp. 323). It is straightforward to show that 
$z(b^1_n(t))$ is a martingale for the modified process. 
$\min\{W^c_{n,1},W^c_{n,m}\}$ is a Markov time with respect to 
the modified process. We apply the optional sampling theorem to conclude that
\begin{equation}
z(n)={\mathbb E}[b^1_n(\min\{W^c_{n,1},W^c_{n,m}\})]=
{\mathbb P}[W^c_{n,1}>W^c_{n,m}]z(m),
\qquad m=n+1,n+2,...
\end{equation}
\end{proof}

\begin{remark}
For small $c$, ${\mathbb P}[W^c_{n,1}<W^c_{n,m}]$ can be expanded into 
a power series in $c$.
\begin{equation}
{\mathbb P}[W^c_{n,1}>W^c_{n,m}]=\frac{(4c)^{m-n}}{[m-2]_{m-n}}
+O(c^{m-n+1}),
\qquad m=n+1,n+2,...
\end{equation}
\end{remark}

Figure 4 shows the hitting probability for $n=10$ and 50 as a function
of $c$. It can be seen that the hitting probability grows quickly around 
a critical value of $c$. For $n=50$ and $m=51$, the probability grows 
linearly until $c$ is smaller than $4.5$ (See Remark 4.4). Then it approaches 
1 quickly.

\section{time to fixation}

In studying evolutionary processes from the standpoint of population genetics,
the probability and the time to fixation of a mutant gene play important 
roles. The expected time to fixation of a mutant gene conditional on fixation 
was obtained by \citet{KimuraOhta1969}. Furthermore, \citet{Ewens1973} and 
\cite{MaruyamaKimura1974} showed that the expected length of time which it 
takes for an allele to increase frequency from $q$ to $y~(>q)$ on the way to 
fixation is equal to the expected length of time which the same allele takes 
when its frequency decrease from $y$ to $q$ on the way to extinction. 
The time-reversibility property is equivalent to the property that the density
of the expected sojourn time does not depend on the sign of the selection 
coefficient, which was shown by \citet{Maruyama1972}. While these results are 
well known, their interpretation in terms of the ancestral process of 
the whole population $\{b_\infty(t);t\ge 0\}$ are interesting.

The fixation probability was obtained by solving the Kolmogorov backward 
equation for the Wright-Fisher diffusion with directional selection
$\{x_p(t);t\ge 0\}$ (\citet{Kimura1957}). The fixation probability of 
the allele $A_1$ is
\begin{equation}
u_1(p)=\frac{1-e^{-4c p}}{1-e^{-4c}},
\end{equation}
and the fixation probability of the allele $A_2$ is $1-u_1(p)$.
It follows from (\ref{fix_A2}) that
\begin{equation}
1-u_1(p)=
2(1-r^2)e^{c(r-1)}
\sum_{k=0}^{\infty}\frac{V_{1k}^{(1)}(c,r)V_{1k}^{(1)}(c,1)}{2\lambda_k}.
\label{id_fix}
\end{equation}
It is possible to obtain the fixation probability from the stationary 
measure of the ancestral process (\ref{stationary}). If the allele $A_2$ 
fixes in a population, the ancestral particles of the whole population
in infinite time backwards consist of type $A_2$ particles only, and
we have
\begin{equation}
{\mathbb E}[q^{b_\infty(\infty)}]=\sum_{i=1}^{\infty}\pi_iq^i
=\frac{e^{4c q}-1}{e^{4c}-1}=1-u_1(p).
\end{equation}
The relationship between the fixation probability and the probability 
generating function of the stationary measure is a special case of 
Theorem 2 (f) in \citet{AS2005}.

The density of time to fixation of the allele $A_2$ conditional on 
fixation has a genealogical interpretation. Let
\begin{equation}
T^c_0:=\inf\{t\ge 0; y_q(t)=1\}.
\end{equation}
Then, it follows from the expression
\begin{equation}
{\mathbb P}[T^c_0<t|T^c_0<\infty]
=\frac{{\mathbb E}[q^{b_\infty(t)}]}{1-u_1(p)}
=\frac{\sum_{i=1}^\infty{\mathbb P}[b_\infty(t)=i]q^i}
{\sum_{i=1}^\infty\pi_iq^i}
\end{equation}
that the process of fixation of the allele in a diffusion model, in which
the left hand side converges to one as $t\rightarrow\infty$, corresponds to
convergence of the distribution of the ancestral process
${\mathbb P}[b_\infty(t)=i]$ to its stationary measure $\pi_i$ as 
$t\rightarrow\infty$. 

The expected time to fixation of the allele $A_2$ conditional on 
fixation was obtained 
by solving the Kolmogorov backward equation 
(\citet{KimuraOhta1969}, \citet{Maruyama1972}), and 
\begin{equation}
{\mathbb E}[T^c_0|T^c_0<\infty]=\int_0^1 \Phi(q,y)dy,
\end{equation}
where $\Phi(q,y)$ is the density of the expected sojourn time of the allele
$A_2$ at frequency $y$ in the path starting from frequency $q$ and 
going to fixation, and
\begin{eqnarray}
\Phi(q,y)&=&\frac{S(y)S(1-y)}{2cy(1-y)S(1)},\qquad y>q,\nonumber\\
&=&\frac{S(y)}{2cy(1-y)}
\left\{
\frac{S(1-y)}{S(1)}-\frac{S(q-y)}{S(q)}
\right\},
\qquad y<q,
\end{eqnarray}
and $S(y)=\exp(4cy)-1$. Then,
\begin{equation}
{\mathbb E}[T^c_0|T^c_0<\infty]=
\int_0^1\frac{S(y)S(1-y)}{2cy(1-y)S(1)}dy 
-\int_0^q\frac{S(y)S(q-y)}{2cy(1-y)S(q)}dy,
\label{exp_T0_diff}
\end{equation}
where
\begin{eqnarray*}
\int_0^1\frac{S(y)S(1-y)}{2cy(1-y)S(1)}dy
&=&\frac{\pi_1}{8c^2}\sum_{i=1}^\infty\sum_{j=1}^\infty
\frac{(4c)^{i+j}}{i!j!}
\int_0^1y^{i-1}(1-y)^{j-1}dy
\nonumber\\
&=&\frac{\pi_1}{2c}\sum_{k=1}^\infty\frac{(4c)^{k}}{(k+1)!}
\sum_{i=1}^k\frac{1}{i(k-i+1)}
\nonumber\\
&=&4\pi_1\sum_{k=0}^\infty\frac{H_{k+1}(4c)^{k}}{(k+2)!},
\end{eqnarray*}
$H_k=1+1/2+\cdots+1/k$,
and
\begin{eqnarray*}
\int_0^q\frac{S(y)S(q-y)}{2cy(1-y)S(q)}dy
&=&\frac{1}{2cS(q)}\sum_{i=1}^\infty\sum_{j=1}^\infty
\frac{(4c)^{i+j}}{i!j!}
\int_0^q\frac{y^{i-1}(q-y)^j}{1-y}dy
\nonumber\\
&=&\frac{1}{2cS(q)}\sum_{i=1}^\infty\sum_{j=1}^\infty
\frac{(4cq)^{i+j}}{i(i+j)!}
{}_2F_1(1,i,i+j+1;q)
\nonumber\\
&=&\frac{1}{2cS(q)}
\sum_{i=1}^\infty\sum_{j=1}^\infty
\frac{(4cq)^{i+j}}{i(i+j)!}
\sum_{k=0}^\infty\frac{(i)_kq^k}{(i+j+1)_k}.
\end{eqnarray*}
It is possible to obtain the expected time to fixation of the allele $A_2$ 
conditional on fixation from the distribution $b_\infty(t)$, and we have
\begin{eqnarray}
{\mathbb E}[T_0^c|T_0^c<\infty]
&=&\frac{1}{\sum_{i=1}^\infty\pi_iq^i}\int_0^\infty t\frac{d}{dt}
\left[
\sum_{i=1}^\infty
\{{\mathbb P}[b_\infty(t)=i]-\pi_i\}q^i
\right]dt
\nonumber\\
&=&
\frac{1}{\sum_{i=1}^\infty\pi_iq^i}
\sum_{i=1}^\infty
q^i\int_0^\infty
\{{\mathbb P}[b_\infty(t)=i]-\pi_i\}
dt.
\label{exp_T0_anc}
\end{eqnarray}
From the two expressions (\ref{exp_T0_diff}) and (\ref{exp_T0_anc}),
an identity at $q=0$ follows immediately.
\begin{equation}
\sum_{k=0}^\infty
\frac{V^{(1)}_{1k}(c,-1)V^{(1)}_{1k}(c,1)}{\lambda_k^2N_{1k}}
=e^{2c}\pi_1^2\sum_{k=0}^\infty\frac{H_{k+1}(4c)^{k}}{(k+2)!}.
\label{id_T0}
\end{equation}
It is straightforward to obtain similar identities by comparing
(\ref{exp_T0_diff}) and (\ref{exp_T0_anc}) in each power of $q$. 
Moreover, explicit expression for the higher moments of the time 
to fixation, conditional on fixation, are available
\citep{Maruyama1972}, and they produce similar identities.

The density of the time to fixation of a single mutant gene conditional
on fixation has interesting properties. Let
\begin{equation}
T^c_1:=\inf\{t\ge 0; x_p(t)=1\}.
\end{equation}
Then, from a time-reversibility argument on the conditional diffusion 
process (\citet{Ewens1973}, \citet{MaruyamaKimura1974}), we have
\begin{equation}
\lim_{q\rightarrow 0}{\mathbb P}[T^c_0<t|T^c_0<\infty]
=\lim_{p\rightarrow 0}{\mathbb P}[T^c_1<t|T^c_1<\infty]
=\frac{{\mathbb P}[b_{\infty}(t)=1]}{\pi_1}.
\label{reverse}
\end{equation}
The same density holds for a mutant gene of allele $A_1$ and a mutant 
gene of allele $A_2$. This property has an intuitive genealogical 
interpretation. The conditional density is given by the probability of 
the whole population being descended from a single real ancestral particle. 
Since there is no variation in the population, selection cannot have 
an effect on it and consequently, the conditional density should not depend 
on the allelic type. 

Theorem \ref{theorem:bound} gives bounds for the density of the time to 
fixation of a single mutant gene conditional on the fixation, and 
\begin{equation}
\frac{1}{\pi_1}-\frac{\eta(\infty;c)e^{-\lambda_0 t}+2c}{\pi_1\lambda_0}\le
\lim_{q\rightarrow 0}{\mathbb P}[T_0^c<t|T_0^c<\infty]
\le 
\frac{1}{\pi_1}-
\frac{\eta(\infty;c)e^{-\lambda_0 t}+2c}{\pi_1(\eta(\infty;c)+2c)}.
\label{bound_fix}
\end{equation}
The bounds are useless when $c$ is large, since the upper and lower bounds do 
not converge as $c\rightarrow\infty$. When $c$ is small, the lower bound is 
sharp. Figure 5a shows the bounds when $c=0.1$. Figure 5b gives 
the conditional density and the derivative of the lower bound 
${\pi_1}^{-1}\eta(\infty;c)\exp(-\lambda_0 t)$. It can be seen that the tail 
of the conditional density is well characterised by the largest eigenvalue.

It is worth noting that the identity (\ref{reverse}) gives the following 
identity in the distribution $b_n(t)$. Its interpretation in terms of 
the ancestral process $\{b_n(t);t\ge 0\}$ is unclear.

\begin{remark}\label{remark:reverse}
(\ref{reverse}) gives
\begin{equation}
{\mathbb P}[b_\infty(t)=1]=
\lim_{n\rightarrow\infty}e^{-4c}\sum_{k=1}^n(-1)^{k+1}{n\choose k}
{\mathbb E}[b_k(t)].
\end{equation}
\end{remark}

\begin{proof}
(\ref{reverse}) is equivalent to
\begin{equation}
\lim_{p\rightarrow 0}\frac{f(p,1;t)}{u_1(p)}
=\lim_{q\rightarrow 0}\frac{f(p,0;t)}{1-u_1(p)}
=\frac{{\mathbb P}[b_\infty(t)=1]}{\pi_1},
\end{equation}
where
\begin{eqnarray}
\lim_{p\rightarrow 0}\frac{f(p,1;t)}{u_1(p)}
&=&\lim_{p\rightarrow 0}\lim_{n\rightarrow\infty}
\frac{{\mathbb E}[x_p(t)^n]}{u_1(p)}
=\lim_{p\rightarrow 0}\lim_{n\rightarrow\infty}
\frac{{\mathbb E}[(1-y_q(t))^n]}{u_1(p)}
\nonumber\\
&=&\lim_{p\rightarrow 0}\lim_{n\rightarrow\infty}
\sum_{k=0}^n(-1)^k{n\choose k}
\frac{{\mathbb E}[q^{b_k(t)}]}{u_1(p)}
=\lim_{n\rightarrow\infty}
e^{-4c}\sum_{k=1}^n(-1)^{k+1}{n\choose k}
\frac{{\mathbb E}[b_k(t)]}{\pi_1}.
\end{eqnarray}
\end{proof}

In the neutral Wright-Fisher diffusion, the density of time to fixation of 
a mutant gene conditional on fixation follows
\begin{equation}
\lim_{q\rightarrow 0}{\mathbb P}[T^0_0<t|T^0_0<\infty]
={\mathbb P}[a_{\infty}(t)=1],
\label{neu_T}
\end{equation}
where $T^0_0$ is the time to fixation of a mutant gene in the neutral 
Wright-Fisher diffusion. From (\ref{exp_T0_diff}), the expected time to 
fixation of a mutant gene conditional on fixation has a simple form
\begin{equation}
\lim_{q\rightarrow 0}{\mathbb E}[T^c_0|T^c_0<\infty]
=4\pi_1\sum_{j=0}^\infty\frac{H_{j+1}(4c)^{j}}{(j+2)!}
<\lim_{q\rightarrow 0}{\mathbb E}[T^0_0|T^0_0<\infty]=2,
\label{exp_T0m}
\end{equation}
where the inequality holds from the following lemma:

\begin{lemma}\label{lemma:sojourn}
The density of expected sojourn time of the allele $A_2$ at frequency 
$y$ in the path starting from frequency $0$ and going to fixation 
satisfies 
\begin{equation}
\frac{S(y)S(1-y)}{2cy(1-y)S(1)}<2, \qquad 0<y<1.
\end{equation}
\end{lemma}

\begin{proof}
The inequality is equivalent to
\begin{equation}
\frac{e^{4cy}-1}{y}\frac{e^{4c(1-y)}-1}{1-y}<4c(e^{4c}-1),
\end{equation}
or
\begin{equation}
\sum_{i=0}^{\infty}\sum_{j=0}^i\frac{(4c)^iy^j(1-y)^{i-j}}{(j+1)!(i-j+1)!}
<\sum_{i=0}^{\infty}\frac{(4c)^i}{(i+1)!}.
\end{equation}
The inequality follows from an inequality
\begin{equation}
\sum_{j=0}^i
\frac{y^j(1-y)^{i-j}}{(j+1)!(i-j+1)!}
<\frac{1}{(i+1)!}\sum_{j=0}^i
\frac{i!y^j(1-y)^{i-j}}{j!(i-j)!}=\frac{1}{(i+1)!},\qquad i=0,1,...
\end{equation}
\end{proof}

As $c$ becomes large, ${\mathbb P}[b_\infty(t)=1]$ decreases, while 
the expected fixation time of a mutant gene conditional on fixation 
decreases. It is straightforward to show that the inequality for 
the expected fixation time (\ref{exp_T0m}) is equivalent to an inequality
\begin{equation}
\int_0^{\infty}
\left\{
\frac{{\mathbb P}[b_\infty(t)=1]}{\pi_1}-
{\mathbb P}[a_\infty(t)=1]\right\}dt>0.
\end{equation}
In the neutral process, both the density of the waiting time until
the ancestral process hits the state 1 and the density of the conditional 
fixation time are given by the probability that the ancestral process is 
at the state 1 (\ref{neu_w},\ref{neu_T}). It follows that
\begin{equation}
{\mathbb E}[W_{\infty,1}^0]=\lim_{q\rightarrow 0}
{\mathbb E}[T^0_0|T^0_0<\infty]
=\int^\infty_0\{{\mathbb P}[a_\infty(t)=1]-1\}dt=2.
\end{equation}
In contrast, in the processes with directional selection, we have
\begin{equation}
{\mathbb E}[W_{\infty,1}^c]=2\sum_{j=0}^{\infty}\frac{(4c)^j}{(j+1)(j+1)!}>2,
\end{equation}
while
\begin{equation}
\lim_{q\rightarrow 0}{\mathbb E}[T^c_0|T^c_0<\infty]
=\int^\infty_0\left\{\frac{{\mathbb P}[b_\infty(t)=1]}{\pi_1}-1\right\}dt
=4\pi_1\sum_{j=0}^\infty\frac{H_{j+1}(4c)^{j}}{(j+2)!}<2.
\end{equation}


\section{discussion}

In this article, the ancestral process $\{b_n(t);t\ge 0\}$, specifying
the number of branches in the ancestral selection graph, was investigated
by exploiting the moment duality between this process and the Wright-Fisher
diffusion with directional selection. An explicit formula for the probability
distribution of $b_n(t)$ was derived. Although this expression cannot be 
given in closed form, since it involves eigenvalues and coefficients which
are determined by an intractable three-term recursion relation, it is
possible to expand the probability distribution as a perturbation series 
in $2c$. This expression is given in a closed form for each order of 
the perturbation and is accurate when $c$ is small. Bounds for the probability
that $b_n(t)$ is at the state 1 is obtained. When $c$ is small, one of
the bounds is sharp. The density of time to fixation of a single mutant gene 
conditional on fixation was shown to be given by the probability of the whole 
population being descended from a single real particle, regardless of 
the allelic type. Thus, the bounds for the probability that $b_n(t)$ is
at the state 1 give bounds for the density of the conditional hitting time.
It was shown that the tail of the conditional hitting time is well 
characterised by the largest eigenvalue. The probability that the process 
hits states larger than the initial state before the process hits the state 1 
was obtained. According to the formula, the number of branches in 
the ancestral selection graph grows rapidly when $c$ is larger than a critical
value. One of the difficulties of simulating the ancestral selection graph 
is keeping track of large numbers of branches when $c$ is large. Specifically,
when $c$ is larger than the critical value, enormous number of branches 
emerge and it will be difficult to simulate the ancestral selection graph.

If a sample consists only of type $A_2$ particles, the probability 
distribution of the ancestral particles, all of which are $A_2$, is $b_n(t)$ 
(see Theorem 2.1). If a sample contains type $A_1$ particles, the joint 
probability distribution of the number of the $A_1$ particles and the number 
of the $A_2$ particles is interesting. However, it seems that the expression
of the moments in the Wright-Fisher diffusion (\ref{dual_xy}) does not give
any insights into the joint probability distribution, except in the case when 
a sample consists of a single $A_1$ particle. The time-reversibility argument 
of the conditional diffusion process gives an identity, whose interpretation 
in terms of the ancestral process is unclear (see Remark 5.1). 
The interpretation of the time-reversibility in terms of the ancestral process 
needs further investigation.

In this article, the fixation process of a mutant gene in the Wright-Fisher
diffusion with directional selection, which has been well studied, was 
interpreted in terms of the convergence of the ancestral process to its 
stationary measure. This approach might be useful for other population 
genetical models for which less is known about the fixation process. Examples 
include diffusion processes with frequency-dependent selection, with multiple 
selected alleles, and of spatially-structured populations \citep{AS2005}. 

\begin{appendix}

\section{The oblate spheroidal wave function}

The oblate spheroidal wave function $V^{(1)}_{1k}(c,z)$ can be represented
by expansions of the form (\citet{Stratton1941})
\begin{equation}
V^{(1)}_{1k}(c,z)=\sum_{l\ge 0}{}'f^k_l(c) T^1_l(z),
\qquad k=0,1,...
\end{equation}
This notation was used in \citet{Kimura1955}. It was denoted by 
$V^{(1)}_{1k}(-ic,z)$ in \citet{Stratton1941} and 
$(1-z^2)^{\frac{1}{2}}S_{1k+1}(c,z)$ in \citet{Flammer1957}.
From the orthogonal properties of the Gegenbauer function it is shown that
\begin{equation}
\int^1_{-1}(1-z^2)V^{(1)}_{1k}(c,z)V^{(1)}_{1l}(c,z)dz
=\delta_{k,l}N_{1k},
\end{equation}
where 
\begin{equation*}
N_{1k}=2\sum_{l\ge 0}{}'\frac{(l+1)(l+2)}{(2l+3)}(f^k_l(c))^2.
\end{equation*}
Note that
\begin{equation}
V_{1k}^{(1)}(c,1)
=\frac{1}{2}\sum_{l\ge 0}{}'(l+1)(l+2)f^k_l(c),\qquad
V_{1k}^{(1)}(c,-1)=(-1)^kV_{1k}^{(1)}(c,1).
\end{equation}
The coefficients $f^k_l(c)$ satisfy a three-term recursion in the form
\begin{equation}
A_{l+2}f_{l+2}^k(c)+B_{l}f_l^k(c)+C_{l-2}f_{l-2}^k(c)=0,
\label{rec_f}
\end{equation}
where
\begin{eqnarray*}
A_{l}=-\frac{(l+1)(l+2)}{(2l+1)(2l+3)},\qquad
B_l=\frac{l(l+3)-b_k}{c^2}-\frac{2l^2+6l+1}{(2l+1)(2l+5)},\qquad
C_{l}=-\frac{(l+1)(l+2)}{(2l+3)(2l+5)},
\end{eqnarray*}
and $b_k=2\lambda_k-2-c^2$. $f^k_l(c)=0$ for odd $l$ if $k$ is even and 
for even $l$ if $k$ is odd. (\ref{rec_f}) can be developed as a continued 
fraction.
\begin{eqnarray}
\frac{f_l^k}{f_{l+2}^k}&=&
-\frac{A_{l+2}}{B_l-}\frac{C_{l-2}A_l}{B_{l-2}-}\cdots
\frac{C_{2}A_{4}}{B_{2}-}\frac{A_2}{B_0}
\qquad l=0,2,...
\nonumber\\
&&-\frac{A_{l+2}}{B_l-}\frac{C_{l-2}A_l}{B_{l-2}-}\cdots
\frac{C_{3}A_{5}}{B_{3}-}\frac{A_3}{B_1}
\qquad l=1,3,...
\label{confra1}
\end{eqnarray}
and
\begin{equation}
\frac{f_{l+2}^k}{f_{l}^k}=
-\frac{C_{l}}{B_{l+2}-}\frac{A_{l+4}C_{l+2}}{B_{l+4}-}
\cdots,\qquad l=0,1,..
\label{confra2}
\end{equation}
$b_k$ is determined by the condition that the reciprocal of the ratio 
$f_{l}/f_{l+2}$ by (\ref{confra1}) must equal the value of $f_{l+2}/f_{l}$ 
obtained from (\ref{confra2}). Then, the continued fractions provide a way 
to compute arbitrary coefficient. 

For small $c$, the eigenvalue can be expanded into a power series in $c$.
\begin{equation}
\lambda_k=\frac{(k+1)(k+2)}{2}+\frac{(k+1)(k+2)}{(2k+1)(2k+5)}c^2+O(c^4).
\end{equation}
If we set $f^k_k(c)=1$, then
\begin{equation}
f^k_{k+2}(c)=
\frac{(k+1)(k+2)}{2(2k+3)(2k+5)^2}c^2+O(c^4),\qquad
f^k_{k-2}(c)=
-\frac{(k+1)(k+2)}{2(2k+1)^2(2k+3)}c^2+O(c^4),
\end{equation}
and other coefficients are zero up to $O(c^4)$. 

\end{appendix}

\vspace{0.5cm}

{\bf Acknowledgments.} The author is grateful for useful discussion with 
Jotun Hein, his group, Robert C. Griffiths and Jay E. Taylor. He is also 
grateful for valuable comments and suggestions by anonymous referees.

\newpage

\noindent
Figure 1. A realization of the ancestral selection graph embedded in 
a diagram of a sample path of $x_p(t)$.

\noindent
Figure 2. $\eta(\infty;c)$ (dots) and $3\exp(-c)$ (line).

\noindent
Figure 3. $1-{\mathbb P}[b_{10}(t)=1]$ (line) and the upper and the lower 
bounds given by Theorem \ref{theorem:bound} (dotted lines).

\noindent
Figure 4. The probability that the ancestral process $\{b_n(t);t\ge 0\}$
hits states $m(>n)$ before the process hits state 1 (ultimate ancestor).

\noindent
Figure 5. (a) The cumulative probability of the density of time to fixation of
a single mutant gene conditional on the fixation (\ref{reverse}) (line) and
the upper and the lower bounds given by (\ref{bound_fix}) (dotted lines).
(b) The density of the conditional fixation time (line) and 
$\pi_1^{-1}\eta(\infty;c)\exp(-\lambda_0 t)$ (dotted line).

\newpage

\begin{figure}[!hbt]
\includegraphics[width=.9\textwidth]{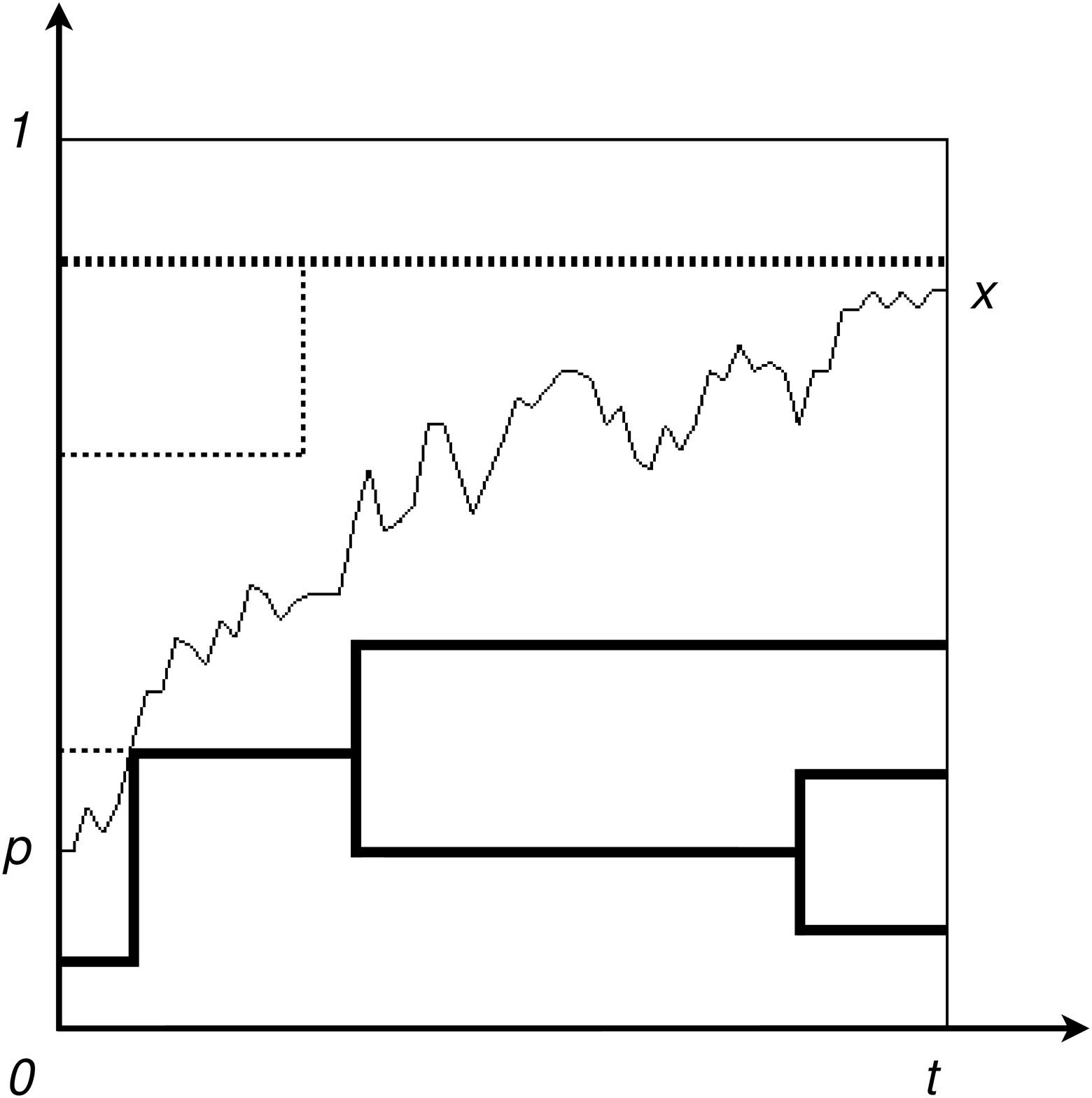}
\caption{}
\end{figure}

\newpage

\begin{center}
\begin{figure}[!hbt]
\psfrag{c}{$c$}
\psfrag{eta}{$\eta(\infty,c)$}
\includegraphics[scale=1.0]{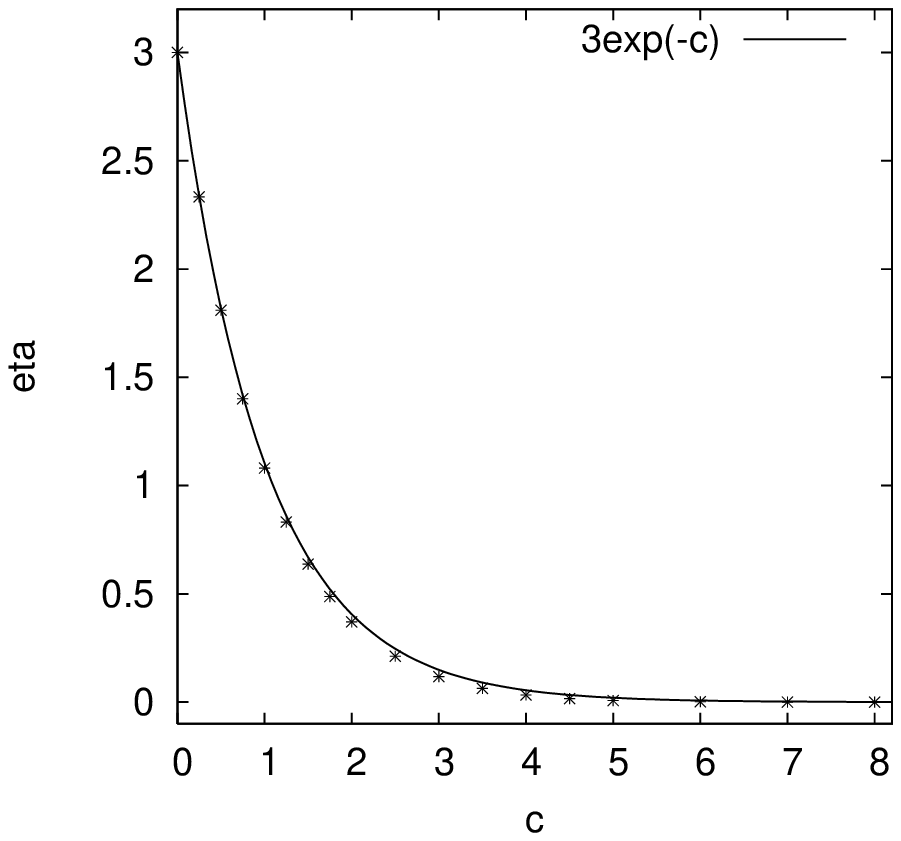}
\caption{}
\end{figure}
\end{center}

\newpage

\begin{center}
\begin{figure}[!hbt]
\begin{tabular}{cc}
\psfrag{t}{$t$}
\includegraphics[scale=0.7]{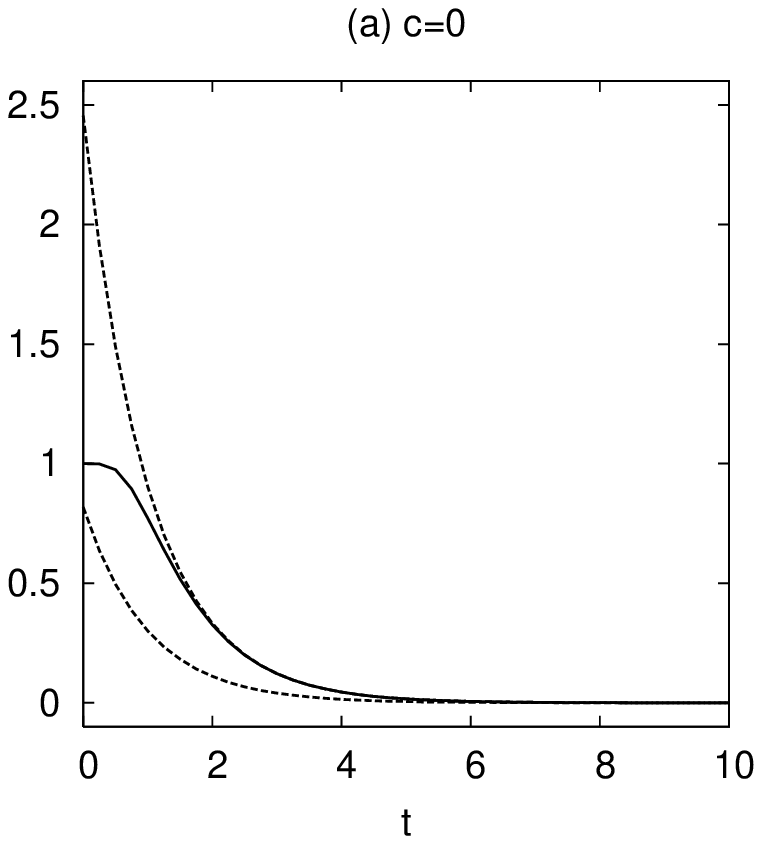}&
\psfrag{t}{$t$}
\includegraphics[scale=0.7]{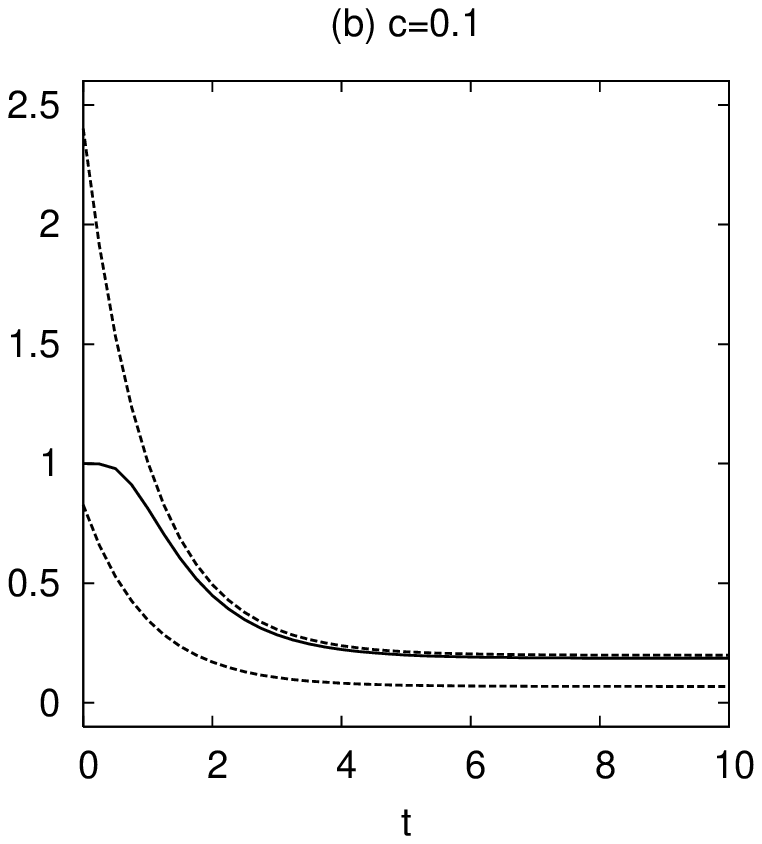}\\
\psfrag{t}{$t$}
\includegraphics[scale=0.7]{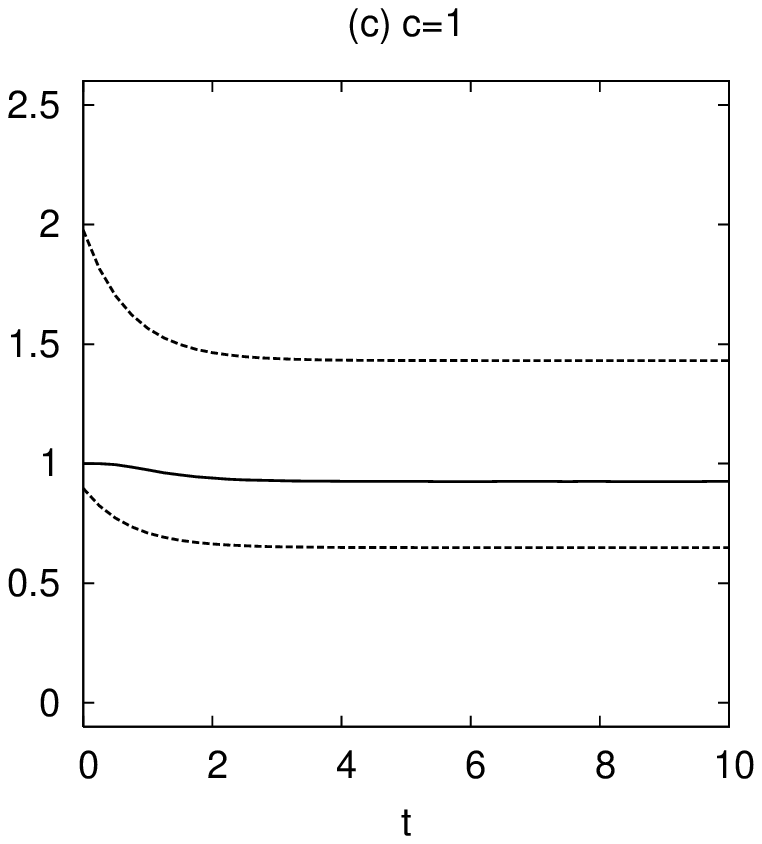}&
\psfrag{t}{$t$}
\includegraphics[scale=0.7]{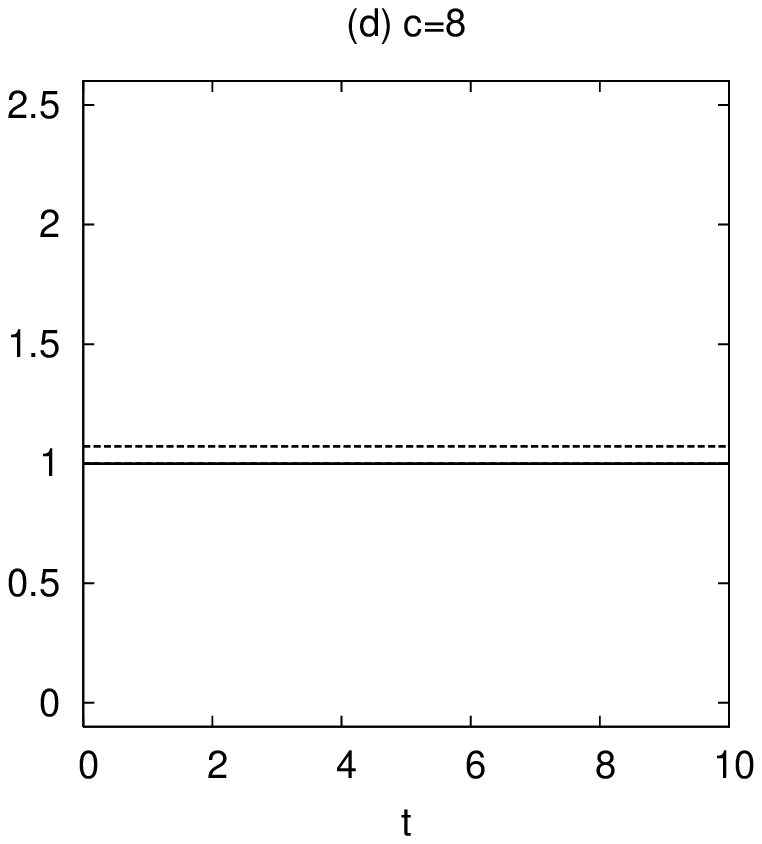}
\end{tabular}
\caption{}
\end{figure}
\end{center}

\newpage

\begin{center}
\begin{figure}[!hbt]
\psfrag{c}{$c$}
\includegraphics[scale=1.0]{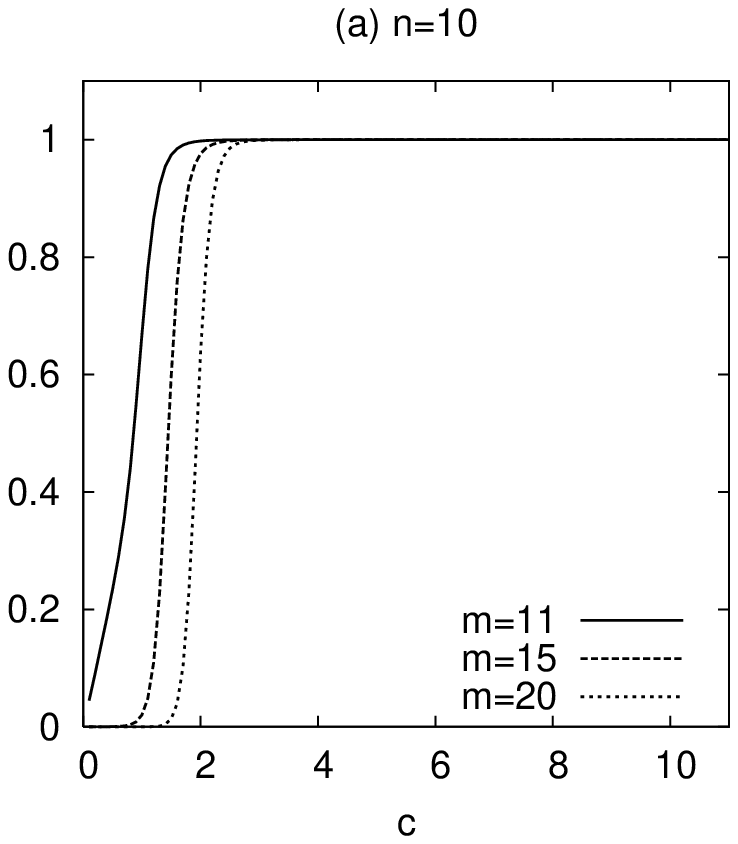}\\
\psfrag{c}{$c$}
\includegraphics[scale=1.0]{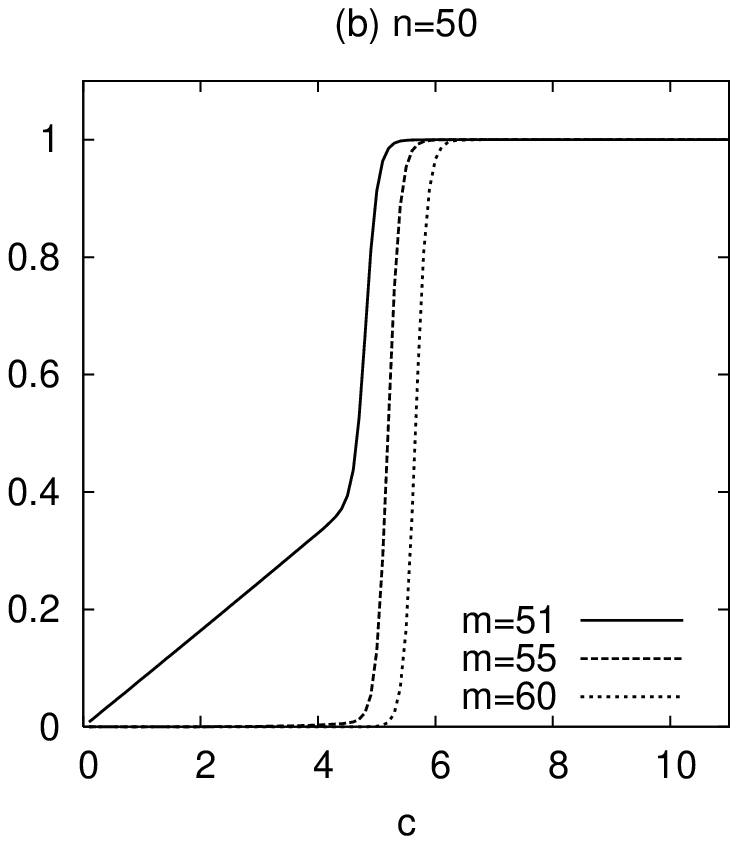}
\caption{}
\end{figure}
\end{center}

\newpage

\begin{center}
\begin{figure}[!hbt]
\psfrag{t}{$t$}
\includegraphics[scale=1.0]{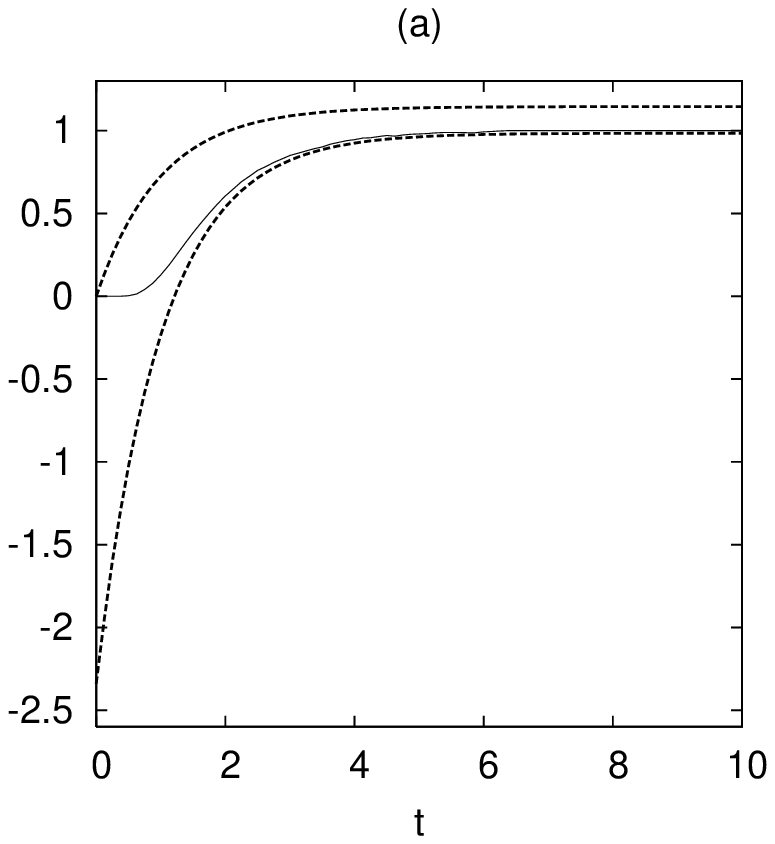}\\
\psfrag{t}{$t$}
\includegraphics[scale=1.0]{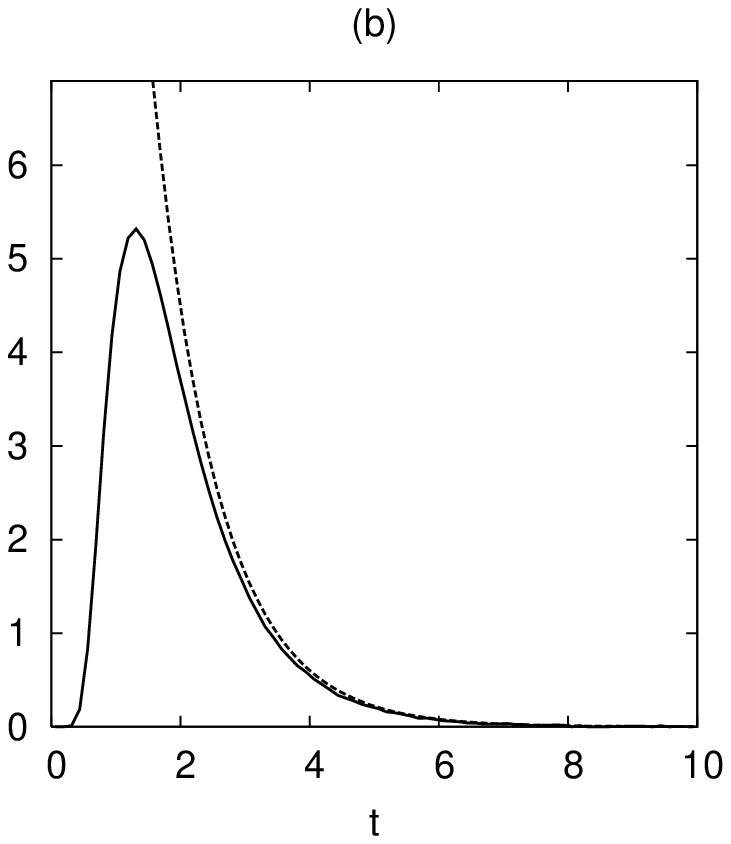}
\caption{}
\end{figure}
\end{center}

\end{document}